\documentclass[conference]{files/IEEEtran}
\pdfoutput=1
\usepackage{hyperref}

\newtheorem{lemma}{Lemma}
\newtheorem{theorem}{Theorem}
\usepackage{graphicx}
\usepackage{amsmath}  
\usepackage{xcolor}   
\usepackage{listings}
\usepackage{graphicx}
\usepackage{subfig}
\usepackage{xspace}
\usepackage{hyperref}
\usepackage{amssymb}

\usepackage{amsfonts}
\usepackage{upgreek}

\usepackage{ragged2e}
\newcommand{\ignore}[1]{}
\newtheorem{definition}{Definition}

\newcommand{\gen}{{\ensuremath{\sf{\mathsf Gen}}}\xspace}
\newcommand{\eval}{{\ensuremath{\sf{\mathsf Eval}}}\xspace}
\newcommand{\verif}{{\ensuremath{\sf{\mathsf Verif}}}\xspace}
\newcommand{\adv}{\ensuremath{\sf{\mathcal Adv}}\xspace}
\newcommand{\negl}{{\ensuremath{\sf{\mathsf negl}}}\xspace}

\newcommand{\confgamename}{{\ensuremath{\sf{\mathsf Game^{QConf}}}}\xspace}

\newcommand{\confgameacro}{{\ensuremath{\sf{\mathsf Game^{QConf}_{\adv, {\Large\uppi}{QLB}}}}}\xspace}

\newcommand{\intgamename}{{\ensuremath{\sf{\mathsf Game^{RInt}}}}\xspace}

\newcommand{\intgameacro}{{\ensuremath{\sf{\mathsf Game^{RInt}_{\adv, {\Large\uppi}{QLB}}}}}\xspace}

\newcommand{\fssgamename}{{\ensuremath{\sf{\mathsf Game^{FSS}}}}\xspace}
\newcommand{\fssgame}{{\ensuremath{\sf{\mathsf Game^{FSS}_{\adv, \Sigma}}}}\xspace}

\usepackage{caption}
\usepackage{subfig} 
\usepackage{graphicx}

\usepackage[linesnumbered,ruled]{algorithm2e}

\ifCLASSOPTIONcompsoc
  \usepackage[nocompress]{cite}
\else
  \usepackage{cite}
\fi
\usepackage[numbers]{natbib}

\begin{document}
\title{ $\uppi$QLB: A \textbf{P}rivacy-preserving with \textbf{I}ntegrity-assuring \textbf{Q}uery \textbf{L}anguage for \textbf{B}lockchain}



\author{\IEEEauthorblockN{Nasrin Sohrabi}
\IEEEauthorblockA{Deakin, Australia \\
nasrin.sohrabi@deakin.edu.au}
\and
\IEEEauthorblockN{Norrathep Rattanavipanon}
\IEEEauthorblockA{PSU, Thailand\\
 norrathep.r@psu.ac.th}
\and 
\IEEEauthorblockN{Zahir Tari\\}
\IEEEauthorblockA{RMIT, Australia\\
  zahir.tari@rmit.edu.au}

}


\newcommand{\acro}{{\large$\uppi$}QLB\xspace}

\maketitle

\begin{abstract}

The increase in the adoption of blockchain technology in different application domains e.g., healthcare systems, supply chain managements, has raised the demand for a data query mechanism on blockchain. Since current blockchain systems lack the support for querying data with embedded security and privacy guarantees, there exists inherent security and privacy concerns on those systems. In particular, existing systems require users to submit queries to blockchain operators (e.g., a node validator) in \emph{plaintext}. This directly jeopardizes users' privacy as the submitted queries may contain sensitive information, e.g., location or gender preferences, that the users may not be comfortable sharing with. On the other hand, currently, the only way for users to ensure \textit{integrity} of the query result is to maintain the entire blockchain database and perform the  queries locally. Doing so incurs  high storage and computational costs on the users, precluding this approach to be practically deployable on common light-weight devices (e.g., smartphones). 

To this end, this paper proposes \acro, a query language for blockchain systems that ensures \emph{both} confidentiality of query inputs and integrity of query results. Additionally, \acro enables SQL-like queries over the blockchain data by introducing relational data semantics into the existing blockchain database. \acro has applied the recent cryptography primitive, namely the function secret sharing (FSS), to achieve the \textit{confidentiality} property. 
To support \textit{integrity}, we propose to extend the traditional FSS setting in such a way that integrity of FSS results can be efficiently verified. Successful verification indicates absence of malicious behaviors on the servers, allowing the user to establish trust from the result. To the best of our knowledge, \acro is the first query model designed for blockchain databases with the support for \textit{confidentiality, integrity}, and \textit{SQL-like} queries.

\end{abstract}


\section{Introduction}

\IEEEPARstart{I}{n} $2008$, Nakamoto \cite{nakamoto2008bitcoin} proposed a new way to build a cash system without the need for having a central trusted authority, such as banks, which marked the generation of cryptocurrencies (e.g., Bitcoin and Ethereum). The innovation lied in building trust among parties where the parties cannot trust each other, yet they all agree on one value of the data. The new technology, blockchain, guarantees the integrity, immutability, and tractability of the stored data. Since then, blockchain has been gaining overwhelming attention and applied in many applications such as, energy trading, healthcare systems, security trading, supply chain management, machine learning (federating learning), and rental systems. A blockchain is defined as a decentralized distributed database that runs over a p-2-p (peer-to-peer) network. It is an append-only data structure (a growing list of ordered records) that is stored among all the peers of the p-2-p network \cite{mohan2017tutorial, sohrabi2020zyconchain, dinh2017blockbench, dinh2018untangling, wang2018forkbase, sohrabi2020scalability}. The ordered records, called blocks, are securely linked together using a cryptography primitive, a hash function, forming a chain of blocks. Peers in the underlying network hosting the blockchain do not trust each other; however, the technology yet guarantees the integrity of the data. This is achieved due to: $(i)$ adding the hash of the previous block into each block, and $(ii)$ the underlying consensus protocol of the blockchain which guarantees the peers store the same replica of the blockchain at any given point of time.

From the database perspective, a blockchain is seen as a decentralised database that stores a large collection of records. Many distributed data-intensive applications have adopted the technology which as a result has increased the demand for querying the newly designed database. For example, in a cryptocurrency application (e.g., Bitcoin), a client may want to query the network to find all the transactions that have occurred during a specific time and the amount of transaction is between a range condition, e.g., $ 20 < Transaction_-Amount < 200 $. In another example, consider a blockchain-based healthcare system, where the data stored on blockchain are sensitive information, such as patients' diseases and records, the doctors who visited the patients, the medications given to the patients, etc. A client may want to query the network to find how many patients with a specific disease have visited a specific hospital over a period of time. Existing solutions to query data records in blockchains have several limitations. Some works, such as FlureeDB \cite{FlureeDB}, IBM \cite{IBM-Blockchain}, Oracle \cite{Oracle-Blockchain} and Bigchain \cite{BigchainDB}, provided a SQL-like query language that relies on a central trusted third party (see Figure~\ref{CentralTTP}). That is, the central server makes a separate copy of the blockchain in a format of a relational database and then enables SQL queries over the data. This approach obviously has a major limitation, i.e., relying on a central trusted server. The server might act maliciously and either does not provide all the results or returns wrong information to the client. This jeopardizes the {\it integrity} of the query results. Additionally, a client's query may contain (and release) sensitive information about the querier. A malicious server may learn about this sensitive information and then use it for its own benefit, compromising the {\it confidentiality} of the user's query. Later, in $2019$ \citet{xu2019vchain} proposed vChain to address the query processing limitations of current blockchain databases. vChain removed the need for having a centralized trusted party to perform queries and addressed the {\it integrity} issue. It enabled a client to verify the result received from the servers. vChian fails however to provide {\it confidentiality}. Additionally, vChain supports only one type of queries, i.e., the Boolean range query (Figure \ref{vChainOverview} depicts a simplistic view of vChain). Hence, the need for having a secure, private and SQL-like query processing model for blockchains remains.

\begin{figure}[h!]
\begin{center}
\includegraphics[width=75mm]{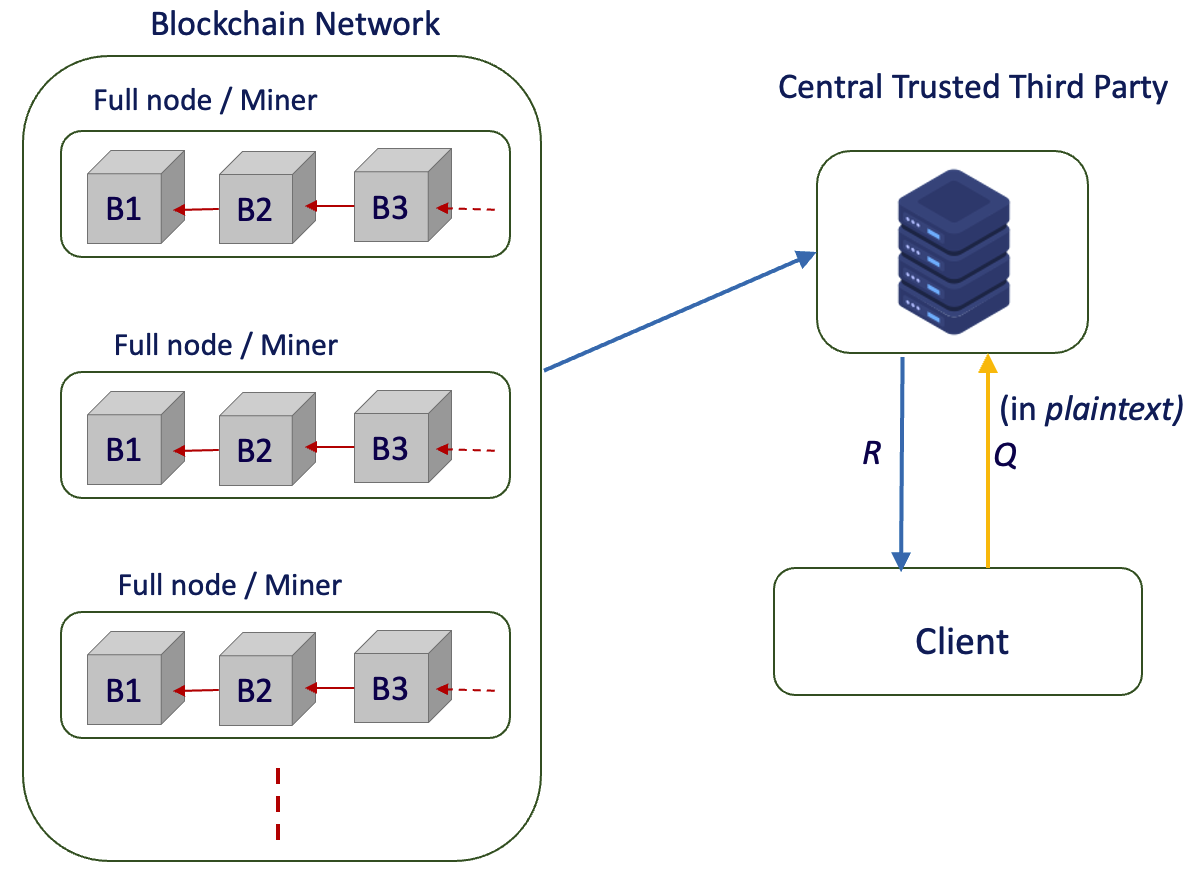}
\caption{An Overview of Query Processing model based on a central trusted third party.} \label{CentralTTP}
\end{center}
\end{figure}

\begin{figure}[h!]
\begin{center}
\includegraphics[width=65mm]{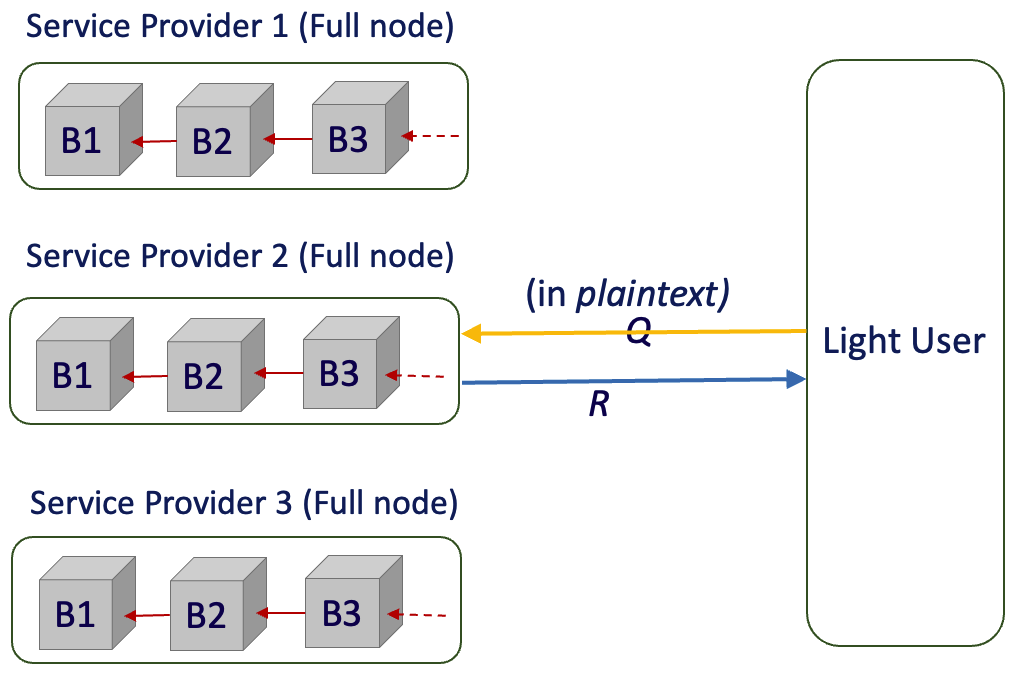}\caption{An Overview of Query Processing Proposed by vChain.} \label{vChainOverview}
\end{center}
\end{figure}

Existing approaches for private query processing have limitations, making them inapplicable for blockchain settings. 
Many private Information Retrieval (PIR) methods \cite{PIR-1995-Chor,xun-2013-PIR, PIR_1998,survey_on_PIR_2007,olumofin2011revisiting,ongaro2014search} assume that a server is honest-but-curious. 
This assumption is rather strong in a public and unreliable domain, such as blockchain, where every node (honest or malicious) can be a server. Additionally,  PIR methods often require multiple round trips to the server to compute the result and therefore consume bandwidth. Other methods based on garbled circuits \cite{bellare2012garblled, goldwasser1997multiparty,yakoubov2017gentle} have high computational and bandwidth costs, making them not applicable for lightweight devices. Encrypted solutions e.g. \cite{popa2011cryptdb, demertzis2020seal, fuller2017sok, papadimitriou2016big,pappas2014blind, popa2014building,mahajan2011depot,li2004secure} are also not applicable for blockchain, as the data on the blockchain is not encrypted. ORAM is another solution proposed to protect data privacy by hiding the access patterns of clients from untrusted servers. The main challenge in ORAM is that a client can download/access only the data that s/he is uploaded/added to the server. This is a different setting in blockchain where everyone can add to the database and clients should be able to retrieve the data added by others. Such work or methods will be explained in details in the related work section, $\S$ \ref{relatedwork}.

To address the aforementioned issues, this paper proposes \acro that aims to: $(i)$ facilitate \textit{SQL-like} queries over blockchain data, $(ii)$ provide \textit{confidentiality} for the queries, and $(iii)$ provide \textit{integrity} for the query results. \acro achieve this by applying the latest cryptography primitive, namely function secret sharing (FSS)~\cite{fss_2015}. Using FSS, \acro divides a client's query into several subqueries and sends each subquery to a miner node. Recall that a typical blockchain network comprises three types of nodes: \textit {full node}, \textit{miner}, and \textit{light node}. The full node maintains a full copy of the blockchain database (block headers and data records), the miner node is a full node that participates in the mining process (generating new blocks), and the light node hosts only the block headers. In \acro, every miner and full node can provide the query processing service to the client, and hence the term \textit{service provider} (SP) is used in this paper to refer to any of such nodes. An overview of \acro is provided in Figure~\ref{Overview-Model}.

\begin{figure}[http]
\begin{center}
\includegraphics[width=85mm]{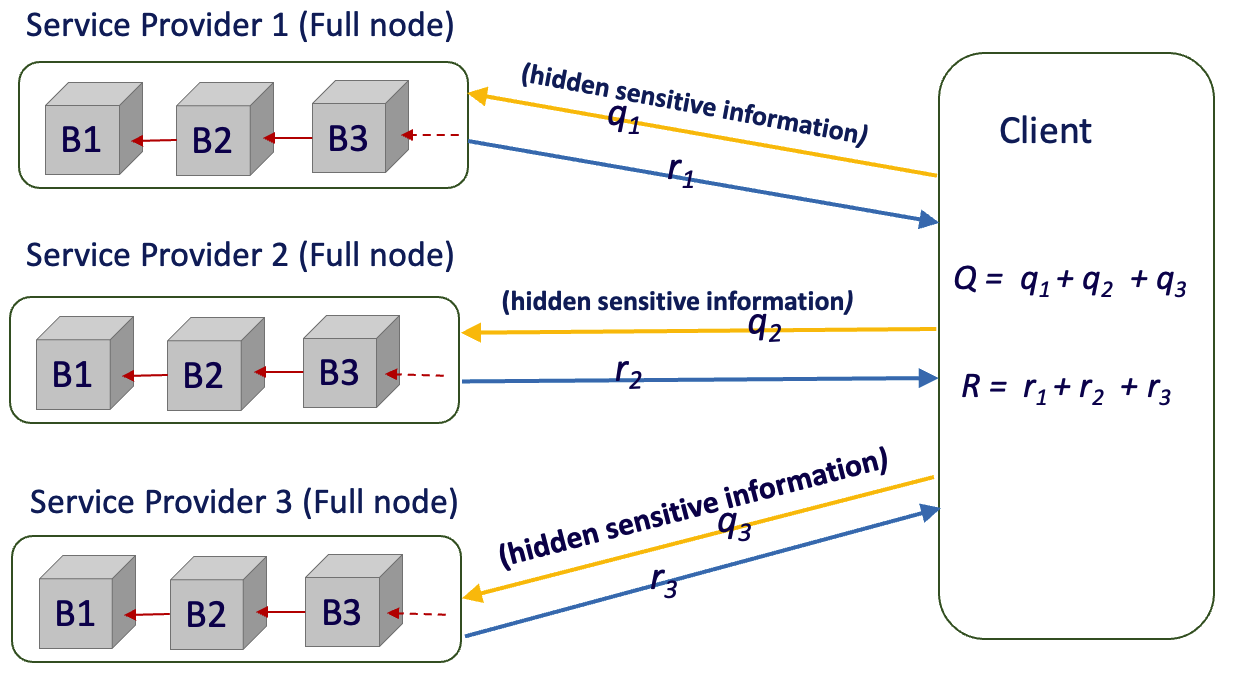}
\caption{\acro Architecture Overview} \label{Overview-Model}
\end{center}
\end{figure}

Consider a common application in a blockchain-based healthcare system that stores records of a tuple of $<T,\ V,\ W>$, where $T$ is the time, $V$ is the cost (it can refer to an insurance number), and $W$ is a set of attributes, $W=<dr_-id, \ pt_-id, \  dr_-spec, \  disease, \ pres, \  hosp, \ loc>$, where $dr_-id$ corresponds to the doctor id (it can be an anonymous id depending on the security requirements of the system), $pt_-id$ is the patient id, $dr_-spec$ is the doctor's specialties, the $disease$ is the disease that has been diagnosed, $pres$ is the treatment that the doctor has prescribed for the patient, $hosp$ is the hospital's name, $loc$ is the geographical location of the hospital. The records are stored in blockchain database which is located in the SPs' premises. A client (e.g., a patient) might want to query the blockchain, e.g., to find all the doctors with specific specialties in nearby hospitals, or to find how many times a specific disease has been diagnosed within a time period in a specific hospital. These queries are submitted to various SPs, and they may however contain sensitive information that clients might not feel comfortable sharing with the SPs. \acro enables clients to hide information based on their specific choice when submitting queries to SPs (\textit{confidentiality}) as well as to prevent malicious SPs from tampering with the query results (\textit{integrity}). 

As depicted in Figure~\ref{Overview-Model}, \acro first divides the client's query into several subqueries. When dividing a query, a client will decide which information needs to be kept hidden (and protected). Each subquery is then sent to a single SP, and the SPs perform the subqueries separately and send the result back to the client.
The client combines the results while checking the integrity. Only when this integrity check passes, it then can obtain the final query output. 
With the assumption that at least one service provider (SP) is honest, \acro is proven to provide \textit{query confidentiality} and \textit{response integrity} security properties. 
It also guarantees that the query is evaluated on all the expected records.

The rest of this paper is organized as follows. Motivating examples are provided in $\S$ \ref{running_example} followed by a description of FSS in $\S$\ref{fss}. The architecture of \acro is detailed in $\S$\ref{architecutre} and the evaluation results are provided in $\S$\ref{evaluation}. Existing work is elaborated in $\S$ \ref{relatedwork} 
 and $\S$\ref{conclusion} concludes the paper.

 \section{Motivating Examples} \label{running_example}
 
 This section provides use-case examples to motivate the importance of security and privacy of query processing in blockchain databases. 
 
 \subsection{Use-case examples}
 
 Consider a blockchain-based stock market where the data records in the blockchain are about the information related to the various stocks (e.g., share prices, available shares, transaction amount, number of shares). In such a system, the SPs will host the information related to the stock market. 

To run transactions (e.g., buy, sell, or check the market prices on specific shares), a client will first need to query the system by sending a request to the SPs. The query will contain information that if the (potentially malicious) SP learns, it can impact the result. For example, if an SP learns that a specific stock (mentioned in a client's query) is either a hot stock (i.e., many clients queried it) or people are bidding on opening price, then the malicious SP could tamper the result and influence the market (e.g., increase the actual price of a specific share). However, if integrity of the response can be verified, the malicious SP will be not able to carry out such an activity.

 In another example, let us consider a blockchain-based rental-car system where car information (e.g., price, model, built-year, color, customer pick-up location, customer drop-off location, city, and rental-period) is the data recorded in the blocks. To rent cars, clients will first perform a search to find their preferred ones. To do so, clients send queries to various SPs. Such queries might however contain sensitive information (e.g.,  pick-up \& drop-off locations, price ranges, or rental-period). If SPs learn such confidential information, it can jeopardize the client's privacy.

\subsection{Query examples}

\acro enables running SQL-like queries over blockchain-based systems. In this section, we describe some examples of such queries. 

Using the rental-car example, we assume that a client wishes to find the minimum price for a specific car model with a specific color. This client also prefers to hide the car model and the color of the car in the query. Hence, the client generates the following query:

\justify{SELECT MIN($price$) WHERE $car_-model$ = ? $\wedge$  $color$ = ?}

Note that ``$?$" indicates the sensitive information that \acro must hide from SPs when sending the query.

Let us consider another example, namely a healthcare system. A client may need to find the total number of patients diagnosed with a specific disease in a specific hospital, or the total number of patients visited a specialist at a specific hospital/location.

Using \acro, the client can hide the information s/he prefers (e.g., $disease$ and $dr_-spec$) by generating the following query: 

\justify{SELECT COUNT($pt_-id$) WHERE $disease$ =? $\wedge$ $dr_-spec$ =?}

\section{Background} \label{fss}

This section provides first some background about function secret sharing (FSS) that serves as the main building block in \acro. Then, we describe the problem stemming from directly applying FSS-based methods to our target setting.

\subsection{FSS Background}

Let $\mathbb{G}$ denote an Abelian group and an integer $p > 1$.
FSS~\cite{fss_2015,fss_ccs_2016,de2022lightweight} is defined by a tuple (\gen, \eval): 
  \begin{itemize}
      \item $\gen(1^{\uplambda}, f) \rightarrow K_1, ..., K_p$: It gets a security parameter $1^\uplambda$ 
      and a function $f : \{0, 1\}^n \rightarrow \mathbb{G}$, then splits $f$ into $p$ function shares (or keys) ($K_1, ..., K_p$).
      \item $\eval(K_i, x) \rightarrow y_i$: It takes $K_i$ and $x$ as input and outputs $y_i$ which is the result corresponding to the party $i$.
  \end{itemize}

FSS provides both correctness and security. FSS's correctness ensures that adding up all outputs from \eval is equivalent to the original $f$, i.e., $\sum_{i=1}^p y_i = y = f(x)$. For security, FSS guarantees any strict subset of $K_i$ does not reveal anything about the original function $f$. The work in~\cite{fss_2015} provides a game-based definition of FSS security in Definition~\ref{def:fss-game}. This game models an adversary (\adv) capable of compromising $p-1$ shares and aiming to use this knowledge to identify which \adv-generated function ($f^0$ or $f^1$) is used to generate these $p-1$ shares. Theorem~\ref{thm:fss} states that an FSS scheme is considered secure if \adv cannot distinguish $f^0$ and $f^1$ based on $p-1$ output shares with more than a negligible probability. In other words, without all $p$ shares, \adv cannot infer anything about the original $f$.

\begin{definition}
	\label{def:fss-game}
	\textbf{FSS Security Game (\fssgamename)~\cite{fss_2015}}:
	
	\begin{enumerate}
		\item Let $\Sigma = (\gen, \eval)$.
		\item \adv selects two functions $f^0, f^1 : \{0,1\}^n \rightarrow \mathbb{G}$ and gives these functions to the challenger.
		\item The challenger samples $b \leftarrow \{0,1\}$ and computes $(K_1, ..., K_p) \leftarrow \gen(1^\uplambda, f^b)$. Then he gives $p-1$ shares, e.g., $(K_1, ..., K_{p-1})$, to \adv.
		\item \adv outputs a guess $b' \in \{0,1\}$. The game outputs $1$ if $b'=b$; otherwise, it outputs $0$.
	\end{enumerate}
\end{definition}

\begin{theorem}
	\label{thm:fss}
	A scheme $\Sigma = (\gen, \eval)$ is secure if for all probabilistic polynomial-time \adv, there exists a negligible function \negl such that:
	
	\begin{equation*}
	Pr[\fssgame(\uplambda) = 1] \le \negl(\uplambda)
	\end{equation*}
	
\end{theorem}

In this work, we consider two types of $f$ for FSS:
\begin{itemize}
    \item Distributed point function $f_{a,y}$ that is evaluated to a fixed element $y \in \mathbb{G}$ on a specific input $a$, and to $0$ otherwise, i.e.:
    
    \begin{equation} \label{one_condition}
        f_{a,y}(x)= 
        \begin{cases}
            y, & \text{if } \   x = a \\
            0,              & \text{otherwise}
        \end{cases}
    \end{equation}
    
    \item Interval function $f_{a,b,y}$ that is evaluated to $y \in \mathbb{G}$ only when an input is within a specific interval $[a,b]$, and to $0$ otherwise, i.e.:
 
    \begin{equation} \label{range_condition}
        f_{a,b,y}(x)= 
        \begin{cases}
            y, & \text{if } \  a \leq x \leq b \\
            0,              & \text{otherwise}
        \end{cases}
    \end{equation}
    
\end{itemize}

For these two types, Theorem~\ref{thm:fss} implies confidentiality of $a$, $b$ and $y$, i.e., their values cannot be inferred from any of the function shares.

\subsection{Using FSS for blockchain} \label{fss_blockchain}

Blockchain can be considered as a public database. It is a decentralized distributed database where the data records are transparent for public, i.e., everyone can access the data. 
To query the public database with privacy-preserving guarantees, several studies have been conducted and different models have been proposed. Recently,  with the introduction of FSS, several researchers proposed different models to privately query the public records~\cite{MinMax_boyle2021fss, fss_ccs_2016, Splinter2017}. 
In such models, a client generates keys using FSS with $ y = 1$ in Eq.\ref{one_condition} or \ref{range_condition}, and sends these shares to the servers. A given server evaluates the function share on each record of the database, sums the evaluation together and returns the result to the client. The client upon receiving all the results adds them up and finds the result of the query.

In all these settings, while FSS guarantees confidentiality of client queries, a client is unable to verify the integrity of the results received from servers.
In particular, a malicious server can introduce an error $\delta$ to the final output without being detected by the client. This can be done by having the malicious server respond to the subquery with $y_{\adv}+\delta$, instead of $y_{\adv}$, where $y_{\adv}$ is a benign output from $\eval(K_{\adv}, x)$. In this case, after aggregating all outputs, the client obtains the incorrect result that is off from the expected value by $\delta$, i.e., $(\sum_{i=1, i \neq \adv}^p y_i) + (y_{\adv}+\delta) = (\sum_{i=1}^p y_i) + \delta = y+\delta = f(x)+\delta$.

In the blockchain context, verifiable integrity of the query response is of paramount importance due to the followings:

\begin{itemize}
    \item Blockchain is a decentralized system, i.e., there is no central authority controlling blockchain; hence, some servers might be malicious. It is therefore important that a client can verify the query response. This is different for public databases, as they are owned by organizations that can protect servers and detect malicious parties at early stages. 

    \item The size of a blockchain database is significantly high and lightweight devices are unable to download the entire database to perform queries locally. Hence, it is important to enable them to verify the query responses. 
   
\end{itemize}

This work describes \acro as a solution to address the shortcomings of existing models. \acro makes use of FSS to enable private search of blockchain databases, while providing the integrity of the query results and retaining the original goal of ensuring the confidentiality of clients' queries.

\section{\acro} \label{architecutre}


\subsection{System Model and Syntax}

As shown in Figure~\ref{fig:architecture_detail}, \acro is comprised of two entities:
\begin{itemize}
    \item {\it Service Provider} (SP) is the full node or the miner in blockchain terminology. An SP retains a full replica of the blockchain database and is able to provide query processing services.
    \item {\it Client} is a user of the system. It sends query requests to SPs. A client can be either a \textit{light node} or an \textit{external node}. As mentioned, a light node in the blockchain terminology refers to the node that only hosts the block header data. However, an external node is the node/user who does not hold any of he blockchain data and does not exist in the blockchain network. It only sends queries to SPs.
\end{itemize}

\begin{figure}[h!]
\begin{center}
\includegraphics[width=85mm]{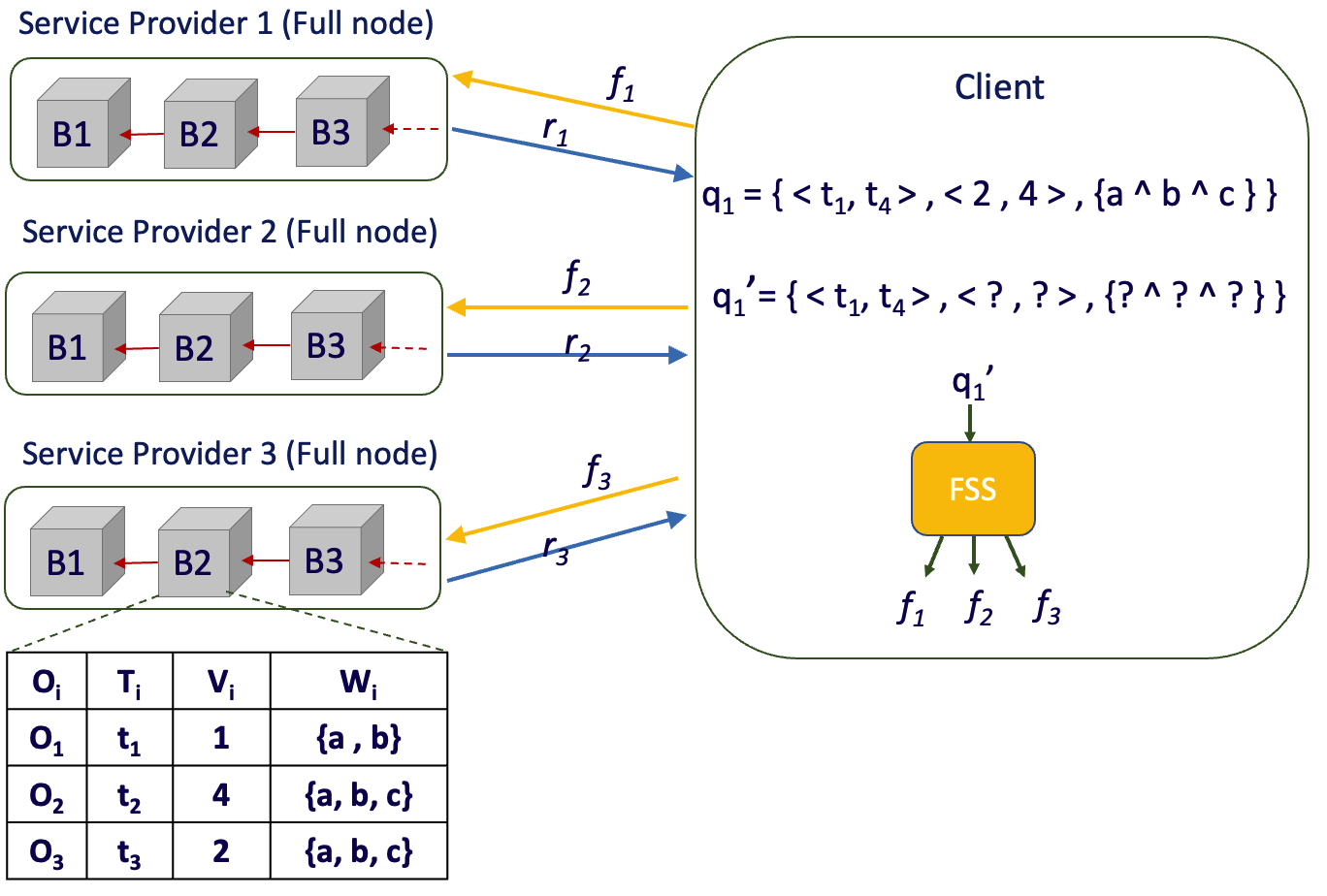}
\caption{\acro Architecture Detail} \label{fig:architecture_detail}
\end{center}
\end{figure}

In terms of algorithms, \acro consists of a tuple $(\gen, \eval, \verif)$, defined as follows: 
\begin{itemize}

    \item $\gen(1^\uplambda, q, s) \rightarrow q', K_1,...,K_p$: It takes as input a security parameter $1^\uplambda$, a query with the format described in $\S$\ref{section:query_format} and a set of secret expressions $s$ in $q$. It then outputs $p$ function shares and a private query $q'$ that removes $s$ value (i.e., replacing it with $?$ as in Section~\ref{running_example}). This function is performed by the client.
    
    \item $\eval(K_i, D, q') \rightarrow y_i$: It takes as input a function share $K_i$, which corresponds to each SP, $D$ -- the entire data from blockchain and a private query $q'$. It then performs an evaluation using $K_i$ and $q'$ on $D$ and outputs $y_i$. 
    This function is performed by a SP.
    
    \item $\verif(y_1, ..., y_p) \rightarrow R$: It takes as input all the share outputs received from the SPs. Then, it verifies whether all $y_i$-s are honestly generated on the same blockchain data. If it detects any malformed $y_i$, it outputs $\perp$ (i.e., aborts). Otherwise, it produces the query result $R$. This function is performed by the client.
    
\end{itemize}

\subsection{Query Model} \label{section:query_format}

 \acro is designed to facilitate the execution of SQL-liked queries on blockchain databases. It first limits the number of blocks that SPs need to search over, then transforms the SQL syntax into a format compatible with the blockchain system. It then uses FSS \cite{fss_ccs_2016} to divide the query into $p$ shares (queries).
 
  Fig~\ref{table:QueryFormat} depicts the \acro's query format, where 
  {\it type} refers to one of the following query types: {\it SUM}, {\it COUNT}, {\it AVG}, {\it MIN} and {\it MAX} queries. 
  {\it blk\_range\_cond} specifies the time range, i.e., $ t_1 < blk_-range_-cond < t_2$, where $t_1$ is the lower bound and $t_2$ is the upper bound. This condition is necessary because in blockchain there is no table to be specified in order to narrow down the search on a database; but rather the full history of all the generated blocks since the genesis block. Hence, it is a must to narrow down the number of blocks that the SP needs to search over; otherwise, the search can be too expensive, time-consuming or may never stop. Thus, \acro shortens the number of blocks by introducing the {\it blk\_range\_cond}, $t_1$ must be bigger than $t_g$, which refers to the time that the genesis block was generated, and $t_2$ must be lower than $t_c$, which is the current time. \acro has considered another condition for {\it blk\_range\_cond} (i.e., the difference between $t_1$ and $t_2$) which should be limited and must not exceed a threshold $\uptau $. This is to ensure that the number of blocks generated between $t_1$ and $t_2$ does not exceed a given threshold. The value of $\uptau$ depends on the block generation rate of the network. If the block generation rate is high, then $\uptau$ must be low; and if the block generation rate is low, then $\uptau$ can be bigger. 
 
\begin{figure}[h!]
\begin{center}
\includegraphics[width=80mm]{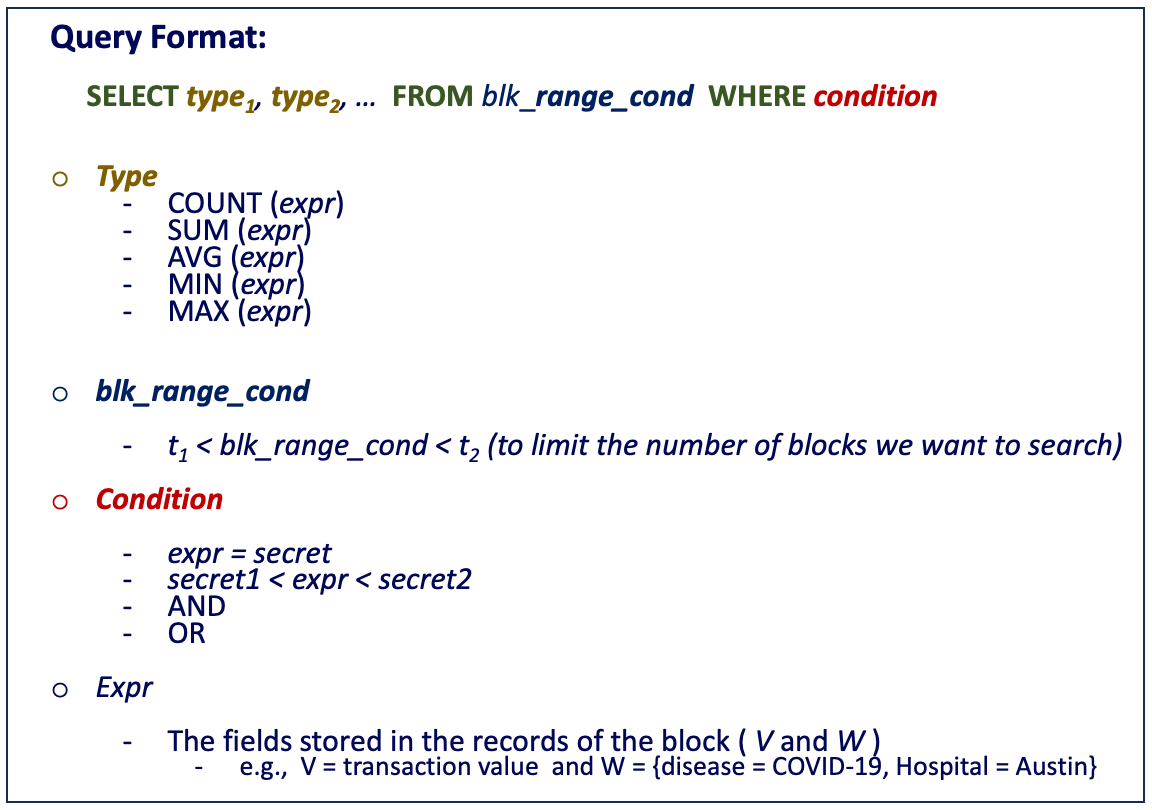}
\caption{Query format in \acro} \label{table:QueryFormat}
\end{center}
\end{figure}

\acro is designed to support two different conditions: $(i)$ {\it single}, and $(ii)$ {\it range} condition.

\begin{itemize}
    \item {\it Single condition} is when the query condition has only one value. That is in the query format where {\it expr} in the condition equals to only one secret value. In this case, the $\gen$ function defines $f$ based on Eq.\ref{one_condition}.
    
    \item {\it Range condition} is when the condition is within a range interval.
    In this condition, the $\gen$ function defines $f$ following Eq.\ref{range_condition}.
\end{itemize}

Additionally, \acro is designed to support two query models: {\it simple} and {\it complex}. The simple model refers to when a query has only one condition (either single or range). The complex model refers to when a query is composed of several conditions using $AND$ and $OR$. Note that, \acro supports not more than one range condition in one query request.

Figure~\ref{fig:architecture_detail} illustrates how a query processing works in \acro, which involves the following steps:
  
\begin{enumerate}
  \item The client generates a query, $q$ and specifies the secret information (i.e., the information that the client prefers not to share with SPs) to specify the private query $q'$.
  \item Using \acro.\gen function, the query is divided into several subqueries/function shares, depicted as $f_i$ in Figure~\ref{fig:architecture_detail}. 
  \item The client forwards the subqueries to SPs.
  \item SPs execute the subquery using \acro.\eval function and return the results back to the client. 
  \item The client combines the results and generates the final output of the query using the \acro.\verif function. This function enables the client to either \textit{abort} or acquire the query's response.  
\end{enumerate}

\subsection{Threat Model}

\acro assumes that at least one of the SPs is honest, i.e., it always honestly follows the protocol and does not collude with other SPs. The rest of the SPs can be malicious. The malicious SP (\adv) may attempt to learn the sensitive information from the client's query in order to set up different attacks. For example, in a blockchain-based stock market application, if \adv finds out the shares that the client is interested in purchasing, \adv may modify the price of the share to charge the client more; or \adv may use the location information to endanger the client. Moreover, \adv can also tamper with the query and the response. Tampering with the query means the adversary may search on a different query instead of the client-provided query. Tampering with the response leads to returning an incorrect or forged response. 

We assume that the client is not malicious. A malicious client is a client that attempts to access private data records or tampers with the existing records in the server. This is inapplicable here because \acro is proposed for the public data records, where the database is not private and not directly modifiable by the client. Hence, everyone can see the data publicly. \acro also assumes that the communication between an SP and the client is over a secure channel (e.g., SSL) and the user of the system is not compromised, i.e., the client follows the protocol. 

\subsection{Security Properties}

At high level, \acro provides two security properties: 
\begin{enumerate}
    \item \textbf{Query Confidentiality}: it ensures that sensitive information from the query remains hidden to SPs.
    
    \item \textbf{Response Integrity}: it guarantees detection when malicious SPs tamper with the query responses.
\end{enumerate}

The Query Confidentiality property is defined using the security game provided in Definition~\ref{def:conf}. It models \adv who can compromise $p-1$ SPs with the goal of learning some sensitive information about the query. It assumes \adv has oracle access to all \acro algorithms. \adv is allowed to submit two queries to the challenger as long as these queries contain the same structure and differ only by the secret value in their query condition. The challenger randomly selects one of the two submitted queries, computes the private query $q'$ and $p-1$ shares from the selected query and outputs them to \adv. \adv wins the game if it can identify which query is used by the challenger.~\\

\begin{definition}
\label{def:conf}
\textbf{Query Confidentiality Game (\confgamename)}:

\begin{enumerate}
    \item Let $\Sigma = (\gen, \eval, \verif)$.
    \item \adv is given oracle access to $\Sigma$ and $D$ blockchain data.
    \item \adv selects a list of secret expressions $s$ and two queries $q_0, q_1$, where these queries differ only by the $s$ value, and gives $s, q_0, q_1$ to the challenger.
    \item The challenger samples $b \leftarrow \{0,1\}$, calls $(q', K_1, ..., K_p) \leftarrow \gen(1^\uplambda, q_b, s)$, 
    and outputs $q'$ and $p-1$ shares (e.g., $K_1, ..., K_{p-1}$) to \adv.
    \item \adv outputs a guess $b' \in \{0,1\}$. The game outputs $1$ if $b=b'$; otherwise, it outputs 0.
\end{enumerate}
\end{definition}

Next, we formalize the Response Integrity property with the security game provided in Definition~\ref{def:int}.
It models \adv that aims is to manipulate the query response by gaining control of $p-1$ SPs.~\\

\begin{definition}
\label{def:int}
\textbf{Reponse Integirty Game (\intgamename)}:

\begin{enumerate}
    \item Let $\Sigma = (\gen, \eval, \verif)$.
    \item \adv is given oracle access to $\Sigma$ and $D$ blockchain data.
    \item \adv selects a query $q$, a list of secret expressions $s$ and sends them to the challenger.
    \item The challenger calls $\gen(1^\uplambda, q,s)$, producing $q', K_1, ..., K_p$. He then gives $q'$ and $p-1$ shares (says $K_1, ..., K_{p-1}$) to \adv.
    \item \adv produces $y_1^{\adv}, ..., y_{p-1}^{\adv}$ and gives them to the challenger.
    \item The challenger computes $y_p = \eval(K_p, D, q')$ and $R' = \verif(y_1^{\adv}, ..., y_{p-1}^{\adv}, y_p)$. He also directly queries $q$ on $D$, producing $R$.
    \item The game outputs $1$ if $R' \neq \perp$ and $R' \neq R$,; otherwise, it outputs $0$.
\end{enumerate}
\end{definition}

Ultimately, we want to prove Theorem~\ref{thm:conf} and Theorem~\ref{thm:int} below 
to formally guarantees Query Confidentiality and Response Integrity, respectively, in \acro.~\\

\begin{theorem}
\label{thm:conf}
\acro provides Query Confidentiality if for all probabilistic polynomial-time \adv there exists a negligible function \negl such that:

\begin{equation*}
    Pr[\confgameacro(\uplambda) = 1] \le \negl(\uplambda)
\end{equation*}
\end{theorem}

\begin{theorem}
\label{thm:int}
\acro provides Response Integrity if for all probabilistic polynomial-time \adv there exists a negligible function \negl such that:
\begin{equation*}
    Pr[\intgameacro(\uplambda) = 1] \le \negl(\uplambda)
\end{equation*}
\end{theorem}

The proofs for these two theorems are provided in Section~\ref{sec:proof} after various concepts are introduced.
 
\section{Query Execution}

For the simplicity we explain the execution of the queries with an example. Assume each block contains several records, each denoted by $O$ which represents an object/transaction id (note that each transaction or object that is generated has an id). And assume each record, say $O_i$, is a tuple of $<t_{i}, V_{i}, W_i>$, where $t$ represents the time that the $O$ is generated; $V$ is a numerical value of $O$, e.g., for a cryptocurrency $V$ is the transaction amount and in a non-cryptocurrency application it may refer to any numerical value required by that system (e.g., the amount of rental price in a blockchain-based rental car system or it could be an index value); $W$ is the set of attributes in the form of $\{Item,\ Price,\ Color\}$. Thus, a block content can be depicted as Figure~\ref{fig:block_content}. \acro executes queries over $W$. For simplicity, we assume a table representation of  $W$, as depicted in Figure~\ref{fig:block_rep}.

 \begin{figure}[http]
\begin{center}
\includegraphics[width=40mm]{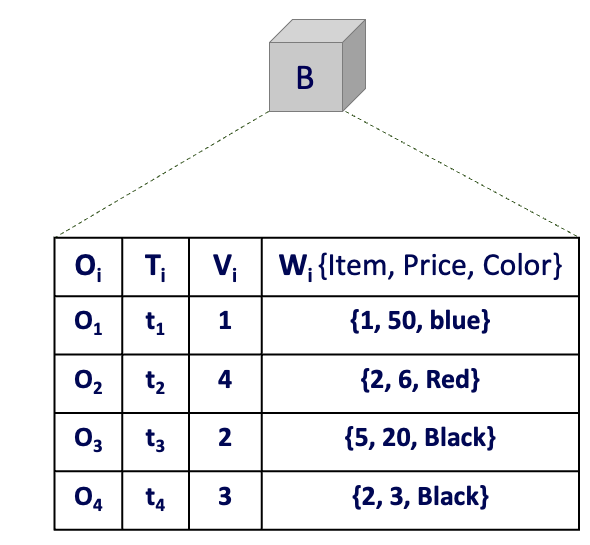}
\caption{Block Content} \label{fig:block_content}
\end{center}
\end{figure}     

 To execute a query, \acro first limits/specifies the number of blocks of the blockchain that it requires to search over based on the time window, denoted as {\it blk\_range\_cond}, which is provided in the query. Here, we assume that the {\it blk\_range\_cond}, returns four blocks (i.e., there are four blocks that their generation time falls between {\it blk\_range\_cond}) and each block contains only one record. Figure~\ref{table} depicts a table representation of the four blocks and their content. 

 \begin{figure}[h!]
\begin{center}
\includegraphics[width=70mm]{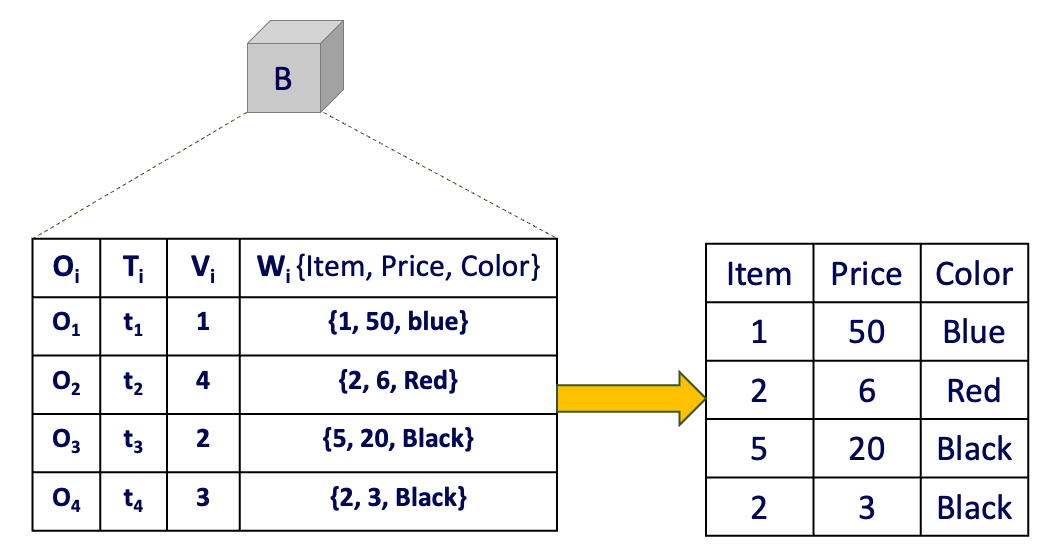}
\caption{Block Representation} \label{fig:block_rep}
\end{center}
\end{figure}

\begin{figure}[h!]
\begin{center}
\includegraphics[width=25mm]{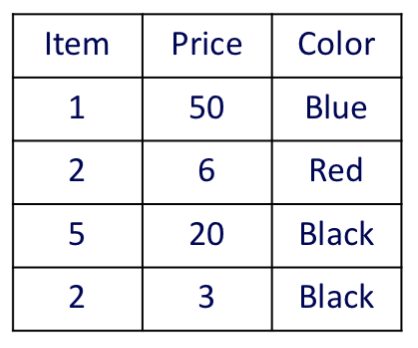}
\caption{Table} \label{table}
\end{center}
\end{figure}

We next exemplify \acro's steps via five queries:

\begin{itemize}
    \item $q_1$: {\justify {SELECT SUM($Price$) FROM $(1/06/2022)$ $<$ $blk_-range_-cond$ $<$ $(4/06/2022)$  WHERE $Item$=$2$}}.
    \item $q_2$: {\justify {SELECT COUNT($Item$) FROM $(1/06/2022)$ $<$ $blk_-range_-cond$ $<$ $(4/06/2022)$  WHERE $4 <Price< 10$}.}
    
\item $q_3$: {\justify {SELECT AVG($Price$) FROM $(1/06/2022)$ $<$ $blk_-range_-cond$ $<$ $(4/06/2022)$  WHERE $Item$ = $2$}.}
\item $q_4$: {\justify {SELECT MAX($Price$) FROM $(1/06/2022)$ $<$ $blk_-range_-cond$ $<$ $(4/06/2022)$  WHERE $Item$ = $2$ $\wedge$ $Color = red$ }}
\item $q_5$: {\justify {SELECT MIN($Price$) FROM $(1/06/2022)$ $<$ $blk_-range_-cond$ $<$ $(4/06/2022)$  WHERE  $2 \leq Item \leq 5$}}
\end{itemize}

Also, note that the $q'$ denotes the query with the hidden information, i.e., the client's sensitive information is hidden in $q'$. For example, $q_{1}'$ can be constructed as follows:

\justify {SELECT SUM($Price$) FROM $(1/06/2022)$ $<$ $blk_-range_-cond$ $<$ $(4/06/2022)$  WHERE $Item$ = ?}

\subsection{Simple Query Model} \label{detail:Simple_Query}

We now discuss how to construct \acro to serve a query with a single condition.

    \item $\gen(1^\uplambda,q,s)$:
    depending on the query condition, {\it single} or {\it range}, this function first samples a random value $y \in \mathbb{G}$ and constructs the $f$ function from $y$, where $f: \{0,1\}^n \rightarrow \mathbb{G}$. In this work, we assume the size of group $\mathbb{G}$ is exponential in terms of $\uplambda$. For instance, $f$ could be $f: \{0,1\}^n \rightarrow \mathbb{Z}_{2^\uplambda}$. Among the above query examples, $q_1$ and $q_3$ are {\it single} and $q_2$ and $ q_5$ are {\it range} conditions. Thus, if we pass $q_1$ and $q_2$ to the $\gen$ function, it behaves as follows:

     \begin{itemize}
         \item {Single condition -- $\gen(1^\uplambda, q_1, \{Item\})$}: \gen takes both $q_1$ and the secret attribute {\it Item} as input. This signifies that {\it Item} is sensitive information and the client aims to hide its value in the query.
         Thus, \gen randomly samples $y$ from $\mathbb{G}$ and defines $f$ as follows:

 \begin{equation} \label{one_condition2}
     f(x)= 
   \begin{cases}
       y, & \text{if } \   x = 2 \\
       0,              & \text{otherwise}
   \end{cases}
 \end{equation}
 
   \item {Range condition -- $\gen(1^\uplambda, q_2, \{Price\})$}: \gen takes as an input the $q_2$ and the secret attribute \textit{Price}.
  Again, it randomly samples $y$ from $\mathbb{G}$ and then defines the $f$ as follows:
 
  \begin{equation} \label{BCQL_range_condition}
    f(x)= 
  \begin{cases}
      y, & \text{if } \  1 < x < 10 \\
      0,              & \text{otherwise}
  \end{cases}
  \end{equation}
\end{itemize}

In this work, we assume that the range condition evaluates to one record only. To allow for multiple records within a specified range, \acro can be extended by adopting the interval evaluation technique used by Splinter~\cite{Splinter2017}, which we leave for future work.  After defining $f$, it then calls FSS.$\gen(1^\uplambda, f)$ to output $p$ function shares. The shares and the private query $q'$ are sent to SPs in a format of $\{ K_i, q'\}$.~\\

 \item $\eval(K_i, D, q')$: 
 the SP, upon receiving $\{ K_i, q'\}$, performs three operations on $D$ (blockchain data). First, it specifies/limits the number of blocks (to process the query) in $D$ according to $blk_-range_-cond$ in the $q'$. Second, it creates an intermediate table depending on the query condition (\textit{single} or \textit{range}, which can be extracted from $q'$). If the query is a {\it single} condition query, the $SP$ executes $GROUP BY$ over the column in the query condition (to recall the query condition, see Figure~\ref{table:QueryFormat}).
 And if the query is a {\it range} condition query, then the $SP$ creates the intermediate table over the column in the query type (to recall the query type, see Figure~\ref{table:QueryFormat}) without executing $GROUP BY$. Third, it generates the binary representation of the values of the column in the query type, denoted as $BIN$. 

 After that, \eval evaluates two sub-functions on the intermediate table: $(i)$ $comp$, stands for compute; and $(ii)$, $merg$, stands for merge. The $comp$ is equivalent to FSS.\eval except it is executed over the binary representation of the values of the column in the query type, $BIN$. That is, $comp$ takes as an input two values: $(i)$ $c$ that refers to the column in the condition, the value of which is hidden; and $(ii)$ $BIN$; and performs the FSS.\eval function over these two inputs.  $comp$ is executed $l$ times, where $l$ is the number of columns of $BIN$, in other word, the maximum number of bits required to represent the binary value of the data, $2^l$. The output of each execution of $comp$ is given to the $merg$ function. $merg$ takes as input all the values generated from $comp$ and merge/concatenate them together and generates a binary array, $y_i$, with $l$ elements. The $y_i$ is the final output of \eval and is returned to the client.  Algorithm~\ref{alg:eval_alg} details the \eval function.
  
\begin{algorithm}
    \DontPrintSemicolon
    $y_i \leftarrow $ [];\\
    $c \leftarrow \text{GetConditionColumn}(D, q')$; \\
    $interTable \leftarrow \text{CreateIntermediaryTable}(D, q')$; \\
    $BIN \leftarrow \text{GetBinaryRep}(interTable, q')$; \\
    $l \leftarrow \text{numColumn}(BIN)$; \\
    \For{$k = 0~\text{to}~l$} 
    {
        $b_k \leftarrow comp(K_i, c, BIN, k)$ ;  \CommentSty{ // see Eq.\ref{com_one_round}} \\
        $y_i \leftarrow merg(y_i, b_k)$; \\
    }
    \Return $y_i$
    \caption{\acro.$\eval(K_i, D, q')$}
    \label{alg:eval_alg}
\end{algorithm}

 The behaviour of this function, specifically the $comp$ sub-function, depends on the query types. We now explain this function in detail for each query type that \acro supports, namely, SUM, AVG, COUNT, MIN, and MAX. 

   \begin{itemize}
       \item SUM: assume that the $SP$ receives $\{K_i, q_{1}'\}$. $SP$ generates the intermediate table for this query, depicted in Figure~\ref{fig:sum_query}. Then, it generates the binary representation of the values of the column $SUM(Price)$, shown in Figure~\ref{fig:sum_query}, as the $BIN$ (recall that the size of each row in the $BIN$ is $l$, i.e., the maximum bits required to represent the binary value of the data). Finally, the \eval function is triggered. \eval function first executes the $comp$ subfunction which takes as an input $Item$ and $BIN$ as in Figure~\ref{fig:sum_query}. One execution of the $comp$ is as follows:
       
        \begin{equation} \label{com_one_round}
          b_k = \displaystyle\sum_{j=0}^{n} FSS.\eval(K_i, c_j) \cdot BIN[j][k] 
        \end{equation}
         This function is repeated for all the columns in $BIN$, Figure~\ref{fig:sum_query}. The outputs of each execution of $comp$, i.e., $b_k$, are sent to the $merg$ subfunction, where the $b_k$-s are merged/concatenated into an array $y_i$. 
         The final $y_i$ is then returned to the client.
         Algorithm~\ref{alg:eval_alg} depicts the \eval function execution.
      
  \begin{figure}[h!]
  \begin{center}
  \includegraphics[width=50mm]{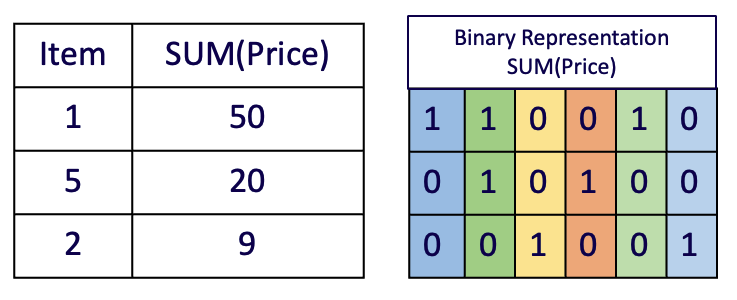}
  \caption{Intermediate Table and the binary representation of the column in query type generated for $q_1$.} \label{fig:sum_query}
  \end{center}
 \end{figure}

   \item COUNT: In this query type, the \eval execution is similar to the SUM-based query type. Assume $SP$ receives $\{ K_i, q_2\prime\}$. It creates the intermediate table and $BIN$ as depicted in Figure~\ref{fig:countquery}. Then, $SP$ can apply $comp$ sub-function for all columns in $c$ and appends the results into a single array $y_i$, which is in turn sent back to the client.
       
\begin{figure}[h!]
\begin{center}
\includegraphics[width=50mm]{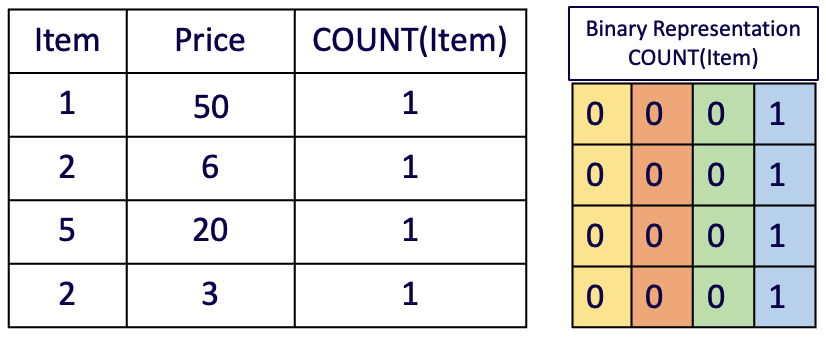}
\caption{Intermediate Table and the binary representation of the column in query type generated for $q_2$.} \label{fig:countquery}
\end{center}
\end{figure}

   \item AVG: This type of query results in a similar execution of \eval as the SUM- and COUNT-based query types. Suppose $SP$ receives $\{K_i, q_3\prime\}$, we refer to Figure~\ref{fig:avgeragequery} for the resulting intermediate table and $BIN$ variable in this query.

\begin{figure}[h!]
\begin{center}
\includegraphics[width=50mm]{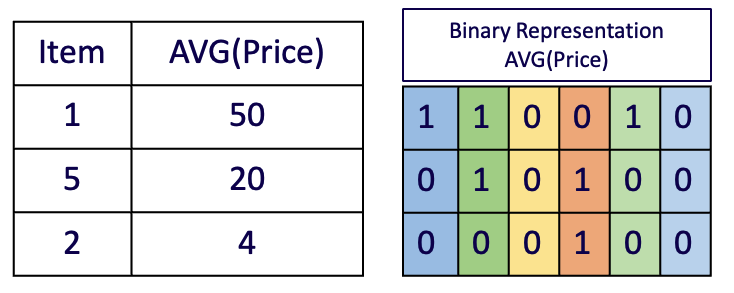}
\caption{Intermediate Table and the binary representation of the column in query type generated for $q_3$.} \label{fig:avgeragequery}
\end{center}
\end{figure}

    \item MAX: among above queries, the $q_4$ is a MAX query which is an example of the complex query model; hence, we defer its explanation to $\S$ \ref{query:complex}.

   \item MIN: if the $SP$ receives $Q_i = \{K_i, q_5\prime\}$, it first creates the intermediate table and then performs the $comp$ and $merg$ similar to the previous queries. We thus omit its detail for brevity.

   \end{itemize}
 
\item \verif$(y_1, ..., y_p)$: Upon receiving the results from all $SP$-s, the client executes the \verif function to output the final results. Recall that $y_i$ is a binary array of the outputs from the \eval function executed by $SP_i$ on each column of the $BIN$ matrix. Thus, in order to find the final result of the query, the client needs to perform the \verif function on each bit of the $y_i$. To perform this, the client follows these steps: $(i)$, it creates a binary matrix of the received $y_i$-s; $(ii)$, it adds the values of each column of the matrix and include it to a $temp$ variable. $if (temp == y)$ then it adds a binary value $1$ into the result array, otherwise $if (temp == 0)$ then it adds a binary value $0$ into the result array; if $temp$ is neither $y$ nor $0$, the client detects a misbehaviour and \textit{aborts} the protocol. This could be due to the server being malicious or other errors in the systems. Algorithm~\ref{verif_alg} details the \verif function.

\begin{algorithm}
    \DontPrintSemicolon
    \For{$i=0~\text{to}~len(y_1)$} {
        $temp \leftarrow \sum_{j=0}^p y_j[i]$;\\
        \CommentSty{// Recall $y$ is selected in \gen}~\\
        \uIf {$temp$ is $y$} {
            $R[i] \leftarrow 1$;
        } \uElseIf {$temp$ is $0$}{
            $R[i] \leftarrow 0$;
        } \Else{
            \CommentSty{// Detect misbehaviour}~\\
            \Return $\perp$; 
        }
     }
    \Return $R$;~\\
    \caption{\acro.$\verif(y_1,...,y_p)$}
    \label{verif_alg}
\end{algorithm}

\subsection{Proof}\label{sec:proof}

This section provides proofs for Theorem~\ref{thm:conf} and Theorem~\ref{thm:int}.

\subsubsection{Query Confidentiality}

Intuitively, \acro inherits the secrecy property from the FSS primitive and therefore \adv cannot infer anything about the secret value in the $f$ function, which is also considered sensitive in \acro.
To provide a more formal argument, we prove Theorem~\ref{thm:conf} via a reduction of \fssgamename in Definition~\ref{def:fss-game} to \confgamename in Definition~\ref{def:conf}. In other words, we aim to show that the existence of \adv that breaks \confgamename can be used to construct $\adv^{FSS}$ that breaks the FSS security guarantee in Theorem~\ref{thm:fss}.

\begin{proof}
    Suppose there exists \adv such that $Pr[\confgameacro = 1] > \negl(\uplambda)$.
    $\adv^{FSS}$ asks \adv to play \confgameacro and can simply use the game output $b'$ to win \fssgamename. This works because we can construct the \fssgamename sequence from \confgamename as follows:
    
    \begin{enumerate}
    	\item \adv can extract $f^0$ and $f^1$ from $q_0$ and $q_1$. Since these queries are sent to the challenger in \confgamename, these functions are also included in this message. (corresponding to Step 2 of \fssgamename)
    	
    	\item The challenger samples $b \leftarrow \{0,1\}$ and invokes \acro.\gen using $q_b$. Since \acro.\gen derives $f^b$ from $q_b$ and calls FSS.\gen on $f^b$, it means that in this step the challenger also computes $(K_1, ..., K_p) \leftarrow FSS.\gen(1^\uplambda, f^b)$ and outputs $p-1$ shares to \adv. (corresponding to Step 3 of \fssgamename)
    	
    	\item \adv outputs $b' \in \{0,1\}$. (corresponding to Step 4 of \fssgamename)
    \end{enumerate}
    
    Therefore, $\adv^{FSS}$ can win \fssgamename with the same non-negligible probability that \adv has of winning \confgamename.

\end{proof}

\subsubsection{Response Integrity}

To prove the response integrity property, we first show \acro satisfies \intgamename in a relaxed scenario, where an intermediate table can be represented using only one bit, hence proving Lemma~\ref{lem:simp}. Then, we can extend this proof to the generic scenario that requires an $l$-column intermediate table, for arbitrary value of $l$.~\\

\begin{lemma}
	\label{lem:simp}
	Theorem~\ref{thm:int} holds for a relaxed scenario where blockchain data $D$ is transformable to a one-bit intermediate table for the input query $q$.~\\
\end{lemma}

\begin{proof}
	In this special relaxed scenario, we consider that after SPs convert $D$ into an intermediate table, this table contains exactly 1 column (or 1 bit). 
	
	The goal of this Lemma is to prove that \adv cannot win \intgamename except with negligible probability in this scenario.
	
	To prove that, we show that even if \adv learns and can manipulate the input query $q$, a list of secret expressions $s$, $p-1$ shares $K_i$ and $p-1$ share outputs $y_i$ (see Definition~\ref{def:int}), he cannot force \acro.\verif (which is performed by the challenger) to produce a valid output without aborting. With a one-column intermediate table, \acro.\eval outputs a single element $y_i \in \mathbb{G}$ and \acro.\verif returns a bit $R$, where:
	\begin{equation}
		\label{eq:simp-sce}
		R =
		\begin{cases}
			0 & \text{if $\sum_{j=1}^p y_j = 0$}\\
			1 & \text{if $\sum_{j=1}^p y_j = y$}\\
			\perp & \text{otherwise}
		\end{cases}  
	\end{equation}
	 We note that $y$ is generated in \acro.\gen and cannot be inferred from $p-1$ shares (due to the secrecy property of FSS in Theorem~\ref{def:fss-game}). It also cannot be influenced by the input query $q$ or the secret expressions $s$ since $y$ is selected at random; in short \adv cannot uncover or manipulate $y$.
	
	Recall that to win \intgamename in this scenario, \adv must output $R'$ such that $R' \neq \perp$ and $R'$ must be different from the benign output produced by an honest execution of the protocol. Without loss of generality, we assume if all SPs are honest, \acro.\eval will output $R=0$; thus to win the game, \adv goal is to produce $R'=1$. Suppose \adv can compromise all SPs except SP$_p$. Then, from Equation~\ref{eq:simp-sce}, to get $R'=1$, 
	\adv must output $y_1^{\adv}, ..., {y}_{p-1}^{\adv}$ such that
	$\sum_{j=1}^{p-1} y_j^{\adv} = y - y_p$. Since $y$ and $y_p$ are random and unknown to \adv, he does not have any other advantage of coming up with $y-y_p$ from $y_1, ..., y_{p-1}$ except for directly guessing one element in $\mathbb{G}$ and hoping that this element corresponds to $y-y_p$. The probability of \adv correctly guessing $y-y_p$ and thus breaking Lemma~\ref{lem:simp} is $\frac{1}{|\mathbb{G}|} = \frac{1}{\exp(\uplambda)} = \negl(\uplambda)$.
	
\end{proof}

To win \intgamename for a generic scenario with an $l$-column intermediate table and hence prove Theorem~\ref{thm:int}, \adv must be able to trick the challenger to execute \acro.\verif and output an $l$-bit $R'$ such that (1) no abort has happened during execution and (2) $R'$ is different from the benign value $R$. However, in order to do that, \adv must get at least one bit from $R'$ to be different from the bit at the same location in $R$. Doing so requires \adv to break Lemma~\ref{lem:simp}. Hence, \adv can win \intgamename with the same negligible probability that \adv has of breaking Lemma~\ref{lem:simp}. $\blacksquare$

\subsection{Complex Query Model}\label{query:complex}

In this section, we explain how complex queries are executed. As mentioned, complex queries are the ones with the combination of $AND$ and $OR$ in the query condition. That is, the query can have more than two conditions. We support the following complex queries:

\begin{itemize}
    
\item AND - Single Conditions: if the complex query is formed by the AND, we concatenate the conditions into one condition and behave the newly formed query as a simple model query. Assume that the query condition is $c_1 = secret_1$ AND $c_2 = secret_2$ AND ... AND $c_n= secret_n$, where each query condition is a single condition. \acro concatenates these conditions into one: $ C = c_1 || c_2 || ... || c_n$. 
This results in a single condition query of the simple model which \acro can process, explained in $\S$\ref{detail:Simple_Query}.

\item OR - Single Condition: if the complex query is formed by the OR and the query conditions are \emph{disjoint}, i.e., $c_1$ does not overlap with $c_2$, then we create $f$ for each condition individually and then add all $f$-s together. Assume that the query condition is $c_1 = secret_1$ OR $c_2 = secret_2 $ OR $...$ OR $c_n= secret_n$. \acro generates $f$ for each $c_i$ using the simple model functions, then combine all the generated functions, i.e, $f = f_{c_1} + f_{c_2} + ... + f_{c_n}$. And finally it passes $f$ to the \acro.\gen function, explained in the simple model above, $\S$\ref{detail:Simple_Query}

\end{itemize} 

Below, we explain an example of a complex query.

 Among the above queries, $q_4$ is a complex query as it contains $\wedge$ in the query condition. Now assume that the $SP$ receives $\{K_i, q_4\prime\}$. In this type of query, compared to the previous queries explained in $\S$\ref{detail:Simple_Query}, $SP$ performs one additional task when creating the intermediate table. It first creates a join column consisting of two columns in the query condition (i.e., for $q_4\prime$, it is $Color$ and $Item$). And then it performs the $GROUP BY$ over this newly created column. Thus, in this example, $q_4\prime$, $SP$ creates the intermediate table as shown in Figure~\ref{fig:maxquery} based on the joint $Color-Item$ columns. Later, it generates $BIN$ (see Figure~\ref{fig:maxquery}) for the column $MAX(Price)$. After the setup is finished, then it performs $comp$ and $merg$ on this intermediate table. The rest of this query is similar to execution of the rest of queries explained in $\S$ \ref{detail:Simple_Query}.
 
 \begin{figure}[h!]
    \begin{center}
     \includegraphics[width=60mm]{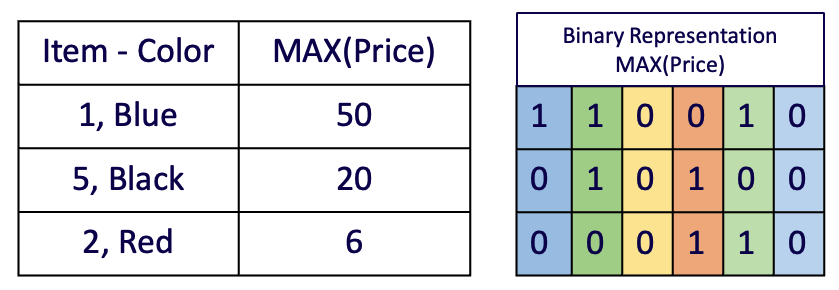}
     \caption{Intermediate Table and the binary representation of the column in query type generated for $q_4$.} \label{fig:maxquery}
      \end{center}
 \end{figure}

\ignore {\oak{Wait, we discussed this before; the solution below does not work... You cannot stop after it finds the result. You need to execute all conditions in the combination in order to get this to work.}
As mentioned, the OR query condition only works when the conditions are disjoint. We propose a naive solution when the conditions are \emph{not} disjoint; however, it is not added to the design of \acro. We believe this could be a good starting point for addressing this problem. The potential solution is as follows: 

If the query conditions are not disjoint, then a table condition can be created that combines all the states that might return true (similar to OR gate) for the query conditions. That is, assume the query comprises of two conditions, $c_1$ and $c_2$. Here, the table is created of all the combination/states of the two conditions, i.e., $2^s-1$, where $s$ is the number of conditions. This results in a $4$ combinations: $(i)$, $c_1 = secret_1 \ ||\ c_2 = X$; $(ii)$, $c_1 = X \ ||\ c_2 = secret_2$; $(iii)$ $c_1 = secret_1 \ || \ c_2 = secret_2$. Then, the first combination of the secrets can be executed by following the steps described in simple mode, $\S$\ref{detail:Simple_Query}. If \verif function outputs a result for this combination, then the client can stop execution here and output the results. However, if the first combination could not find any result in the $D$, then the system needs to execute the second combination. This needs to continue until either \verif outputs a result or all the combinations are executed.  }

\section{Evaluation} \label{evaluation}

To evaluate \acro, we compared its performance, in terms of \textit{computation cost} ($\S$\ref{computation_comp}), \textit{bandwidth} ($\S$\ref{bandwidth_comp}), and \textit{latency} ($\S$\ref{latency_comp}), with two baselines: Splinter and Waldo. We selected these baselines as they both applied the FSS to provide query's \textit{confidentiality}, similar to \acro. 

\subsection{Computation Comparison} \label{computation_comp}
Computation cost/complexity refers to the computations required to process the query and retrieve the data. This cost is $O(N \times log N)$ for Splinter, where $N$ is the number of database records. \acro requires less asymptotic time, i.e., $O(N \times l)$, where $l$ is a small constant indicating the number of bits required to represent each database column. Whereas Waldo's computation cost grows linearly with the number of query predicates: $O(N \times p)$.

We calculated the computation complexity of \acro and then plotted its behaviour for different number of data records in the database. We then plotted the behaviour of Splinter and Waldo over the same set of data records used for \acro. For Waldo and Splinter we used the computation complexity provided in the original papers. As depicted in Figure~\ref{fig:computation_cost}, \acro performs slightly better compared to the baselines when the number of records in the database are small. However, when the record size increases, the performance of Waldo and Splinter deteriorates significantly compared to \acro. More specifically, when the number of data records increases from $10^{8}$ to $10^{9}$, we can see that \acro performs faster (a magnitude of two) compared to Splinter. This difference is even more when we compare \acro and Waldo, where for $10^{9}$ number of records \acro performs over four times faster than Waldo (with $2 -predicates$). When the number of predicates increases Waldo's performance degrades the most compared to Splinter and \acro.

\begin{figure}[h!]

\begin{center}
\includegraphics[width= 90mm]{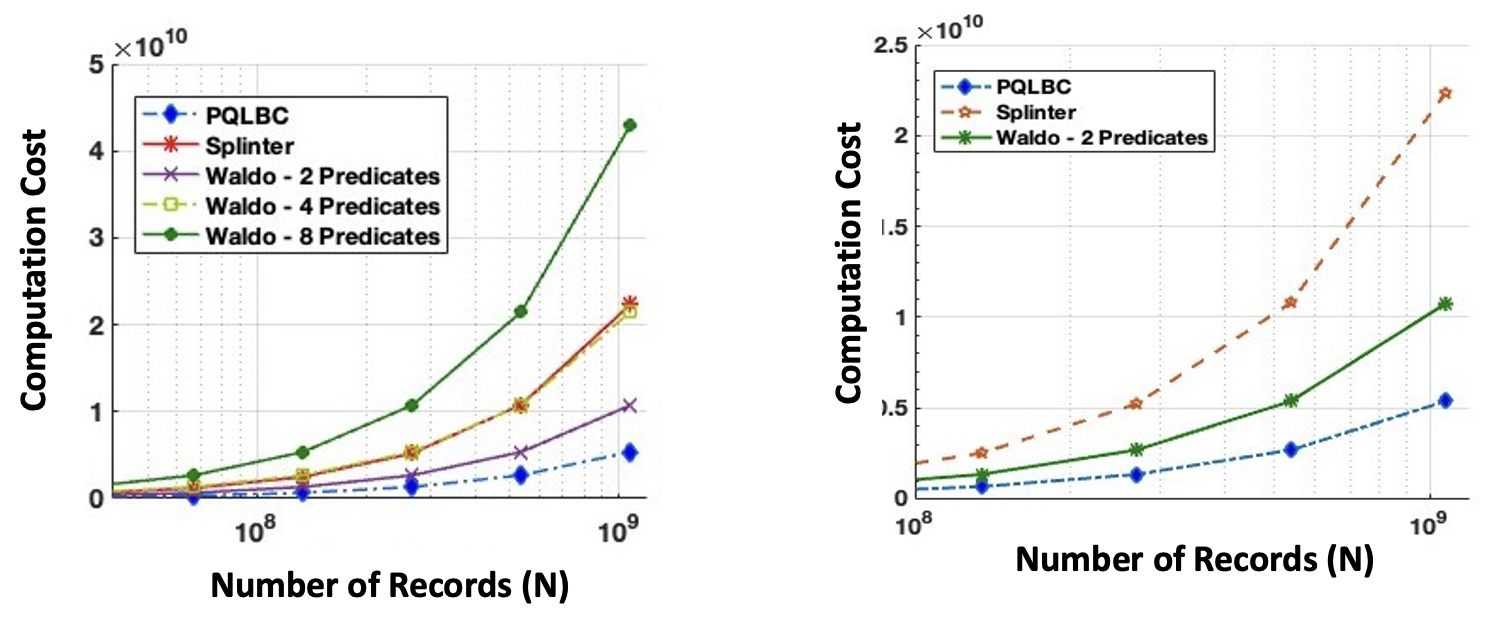}
\caption{Computation cost comparison} \label{fig:computation_cost}
\end{center}
\end{figure}

\subsection{Bandwidth Comparison} \label{bandwidth_comp}

After we evaluated the computation cost of \acro, we switched our attention to the bandwidth comparison. We calculated the client and server bandwidth separately and compared the \acro's results with the baselines. Figure~\ref{fig:bandwidth} illustrates the evaluation results. For the client bandwidth, we compared \acro with Waldo only as Splinter did not provide enough information on the paper to allow us to compute the client's bandwidth. The results show that in \acro, a client requires slightly less bandwidth to perform queries compared to Waldo. This is due to the fact that Waldo generates two pairs of FSS keys for each SP. That is, when the number of SPs are two, Waldo's client generates in total $four$ pairs of keys to send to the SPs. While in \acro client only generates one key for each SP. 

  \begin{figure} 
  \begin{center}
\includegraphics[width= 90mm]{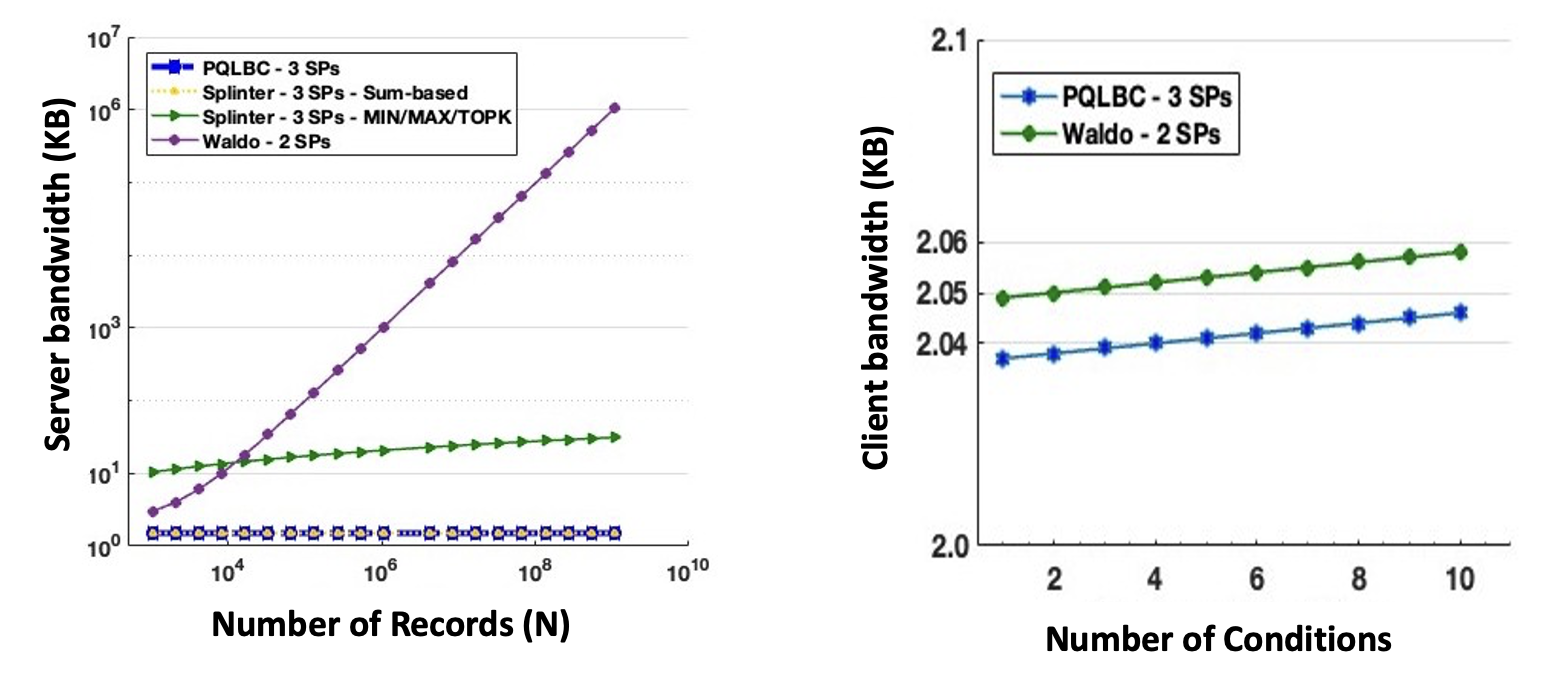}
\caption{Client and server bandwidth comparison} \label{fig:bandwidth}
\end{center}
\end{figure}

However, the server's bandwidth is computed based on: $(i),$ the round trips between client and server, $(ii),$ the size of the data that the server sends to client, and $(iii),$ the Maximum Transmission Unit, $MTU$, each server has. We assume all the servers are running on AWS (Amazon Web Services) EC2 instances and each instance type is $m5.xlarge$ with $10\ Gbps$ network bandwidth and the $MTU$ is $1500$ bytes. As illustrated in Figure~\ref{fig:bandwidth}, the server's bandwidth in \acro is constant and independent from the number of records. This is exactly similar in Splinter for sum-based queries. However, Splinter requires $10$ times more bandwidth for TOPK, MIN, and MAX queries compared to \acro. This is due to the extra communication between server and client for building the intermediate table in Splinter. Waldo, however, requires the highest bandwidth compared to \acro and Splinter as it grows linearly when the number of records increases.

\subsection{Latency Comparison} \label{latency_comp}

Last, we evaluated the latency. We assumed the minimum latency between client and the server (SP) for each message is $1 \ ms$. To plot the results, we computed the latency for one query over different size of records and repeated it for \acro, Splinter, and Waldo (with $2, 4,$ and $8$ predicates). Note that the number of predicates does not affect the latency of the \acro and Splinter. However, the latency of Waldo varies depending on the number of predicates. As illustrated in Figure~\ref{fig:latency}, \acro has the lowest latency compared to the baselines. More specifically, when the number of records are small (below $2^{27}$), \acro, Splinter, Waldo-2-predicates, and Waldo-4-predicates have almost the same latency, less than $2$ seconds. However, when the number of records increases,  the latency surges significantly for Splinter and Waldo, while \acro experiences a moderate increase in latency. For example, when the number of records is increased to $2^{28}$, \acro's latency is still $1$ seconds; however, Splinter's and Waldo-4-predicates's latency surges to around $5$ seconds. For the maximum records of $2^{30}$, \acro's latency is around $5$ seconds which is a factor of $2$ and $4$ times less than Waldo-2-predicates, and Splinter and Waldo-4-predicates, respectively. Waldo-8-predicates experiences the highest latency. 
 
 \begin{figure}[h!]
\begin{center}
\includegraphics[width= 40mm]{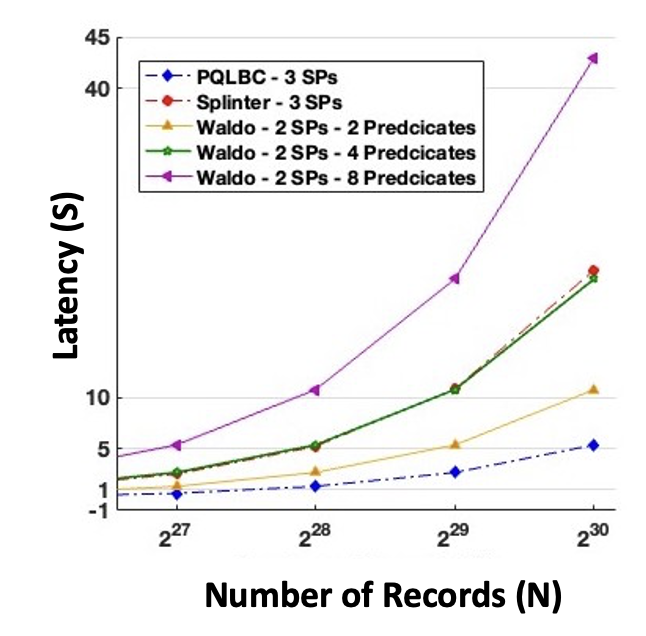}
\caption{Latency comparison} \label{fig:latency}
\end{center}
\end{figure}

\section{Related work} \label{relatedwork}

As previously elaborated,  \acro's main goal is to provide query {\it confidentiality} and response {\it integrity} for performing search queries over a blockchain database. There has been a large body of studies conducted in a similar area (i.e., secure and private search over public data) and this section aims to describe such studies and put our work context.

Related work can be classified into three categories: $(i)$, those providing {\it query confidentiality}, $(ii)$, those enabling queries over {\it encrypted database} to protect the stored data, $(iii)$ those enabling queries over {\it blockchain database}. Compared to the work in $(i)$, \acro provides integrity for the query responses, improving the security by a large factor. The second classification aims to solve a different problem compared to \acro, as they protect private and sensitive data from a malicious server while providing a search function. And the final classification is very closely related to \acro, as \acro is proposed for blockchain databases. 

\subsection{Queries on database with query confidentiality}

\subsubsection{ PIR Systems \cite{PIR-1995-Chor, xun-2013-PIR, PIR_1998, survey_on_PIR_2007,olumofin2011revisiting,ongaro2014search}}
In $1995$, \citet{PIR-1995-Chor} introduced private information retrieval (PIR) that allows a user to retrieve an element of a database in a setting where the server is untrusted (i.e., the client does not want to reveal to the server what information s/he is retrieving). In other word, the goal of PIR is to enable querying the $i$-th item of a database with $n$ items, without revealing $i$. 
 
The introduction of PIR led to the generation of two new schemes: $(i)$ {\it information theoretic PIR} (IT-PIR), also known as multi-server PIR~\cite{PIR-1995-Chor, demmler2014raid, PIR_2012_Optimally, goldberg2007improving_PIR}; and $(ii)$ {\it computational PIR} (CPIR), also known as single-server PIR~\cite{brakerski2014efficient_homomorphic_enc,cachin1999computationally_PIR, kushilevitz1997replication_PIR, xun-2013-PIR, melchor2016xpir, SealPIR2018_S&P}. The IT-PIR scheme comprises of more than one server and each server hosts a replica of the database. The client sends a query to all the servers and combines all the responses received from the servers. IT-PIR relies on a strong security assumption; i.e., the servers are honest and do not collude together. 
Meanwhile, a limitation of CPIR relates to the high computational cost of the cryptographic operation on each element of the database.



\subsubsection{Splinter \cite{Splinter2017}}

 In $2017$, \citet{Splinter2017} applied FSS~\cite{fss_2015} cryptographic primitive to build a new query scheme, called {\it Splinter}, that facilitates private SQL-like queries over a database. The database is public, i.e., is accessible by everyone to read (but not to write) and it is based on a centralized setting (i.e., it is owned/controlled/modified by one organization/company). Splinter supports various types queries, such as SUM, COUNT, AVG, STDEV, MIN, MAX, and TOPK queries. Although it has achieved significantly better performance compared to previous private information retrieval systems, thanks to the use of FSS, it has some limitations: $(i)$ it does not support the {\it integrity} of the query response. That is, Splinter assumes that the service providers are honest (i.e., faithfully following the protocol); $(ii)$ Splinter is designed for the centralized settings, i.e., the database is controlled by one organization. Hence, it cannot be applied in the context of the public blockchain.

 Splinter is closely related to \acro as they both have applied FSS cryptographic primitive to facilitate private search over databases. Nonetheless, \acro additionally addresses the aforementioned limitations, making query processing applicable to the blockchain setting. 
 
\subsubsection{Waldo~\cite{dauterman2022waldo}} It is a private time-series encrypted database that supports multi-predicate filtering while hiding the query content and search patterns. Waldo used FSS to achieve their security guarantee. It relies on the majority honest parties, and has defined up to $three$ parties. This means that they allow only one malicious party. They have achieved a better security compared to their related work based on time-series databases. 
We considered Waldo to be related to \acro as both of them are based on FSS. Nonetheless, Waldo targets a private setting, where the user has a full control over the data on each service provider, which is not applicable to the public setting in blockchain.

\subsubsection{Garbled circuit~\cite{bellare2012garblled, goldwasser1997multiparty,yakoubov2017gentle}} 
It is a cryptographic primitive that provides secure multiparty computations. Embark \cite{lan2016embark} applied this method to design a system that enables a cloud provider to support the outsourcing of network processing tasks (e.g., firewalling, NATs, web proxies, load balancers, and data exfiltration systems), similar the way that compute and storage are outsourced; while maintaining the client’s {\it confidentiality}. To achieve this, Embark first encrypts the traffic that reaches the cloud and then the cloud processes the encrypted traffic without decrypting it. Another work, namely BlindBox \cite{sherry2015blindbox}, also used Garbled circuits to perform deep-packet inspection directly on the encrypted traffic. BlindBox is suggested to be a good candidate for applications such as intrusion detection (IDS), exfiltration detection, and parental filtering. These works mainly have used garbeled circuits to perform private communication on a single untrusted server. The main issue with these proposals is the high computation and bandwidth costs for queries when the datasets are large as a new garbled circuit needs to be generated for each query.

\subsubsection{ORAM}

Oblivious RAM (ORAM) is a cryptographic primitive that protects the privacy of clients by hiding the memory access patterns seen by an untrusted storage server. It basically eliminates the information leakage in memory access traces. In an ORAM scheme, a client
can accesses data blocks that are hosted on an untruseted server. However, for any two logical access sequences of the same length, the communications between a client and a server cannot be distinguished. Several work \cite{lorch2013shroud,stefanov2018path_ORAM, ORAM_ren2015constants} have applied ORAM to develop a secure access model to stored data. However, ORAM is not applicable in the settings where a client is resourced constrained, i.e., it does not have enough memory to be able to load a large volume of data from the server. Additionally, ORAM allows clients to access/download the data that they have written on the server. Hence, it is not suitable for the case where the database is public where everyone can write/read to/from it.

\subsection{Queries on encrypted databases}

 Several prior work~\cite{popa2011cryptdb, demertzis2020seal, fuller2017sok, papadimitriou2016big,pappas2014blind, popa2014building,mahajan2011depot,li2004secure} proposed methods to run queries over an encrypted database. Data in such systems is encrypted before stored, and queries are executed over encrypted data. The aim here is to protect private data from an untrusted or a compromised server. Hence, the purpose of such systems is different from the one of \acro, which aims to protect the confidentiality of queries over a public database, where servers could be malicious and the database is publicly accessible. 

\subsection{Queries on blockchain}

 Lastly, current search approaches for blockchain databases require users to maintain the entire blockchain to run queries (in order to support query integrity). This is not an economical approach as the maintenance cost of blockchain is significant and blockchain’s data size is considerably large. Many companies, including IBM \cite{IBM-Blockchain}, Oracle \cite{Oracle-Blockchain}, FlureeDB\cite{FlureeDB}, and BigchainDB \cite{BigchainDB} provided a search functionality for blockchain databases that unfortunately rely on a central trusted party that executes user’s queries. This obviously causes a single point of failure and {\it confidentiality} as well as {\it integrity} guarantees are not also supported.

 In $2019$, \citet{xu2019vchain} proposed vChain to address a limitation in the previous search schemes for blockchain. vChain's main focus was to provide the {\it integrity} for the search without requiring nodes to download the entire blockchain database. It enabled verifiable Boolean range queries for blockchain databases, where clients are able to verify the results of their searches. Although vChain did achieve a good search performance and provided integrity, it did not consider the {\it confidentiality} of clients' queries.  
 
VQL \cite{wu2021vql} is another work that has been proposed for search improvement on blockchain databases. To provide both efficient and verifiable data queries, VQL adds a Verifiable Query Layer to the blockchain systems. This middleware layer extracts transactions stored in the underlying blockchain system and reorganizes them in databases to provide various query services for public users. It has also applied a cryptographic hash value for each constructed database to prevent data being modified. The main issue with this approach is clearly the extra layer added to the blockchain system. This directly impacts the blockchain size.
Additionally, it degrades the overall performance of the blockchain systems as reorganizing and creating the extra layer incur an overhead on top of the tasks the blockchain systems need to perform. 
More importantly, VQL does not provide the query {\it confidentiality}.
 
\section{Conclusion} \label{conclusion}

  This paper described \acro, a query language for blockchain systems that ensures \emph{both} confidentiality of query inputs and integrity of query results. To achieve \textit{confidentiality} property \acro has applied the function secret sharing primitive (FSS). And, to support \textit{integrity}, we extended the traditional FSS setting in such a way that integrity of FSS results can be efficiently verified. \acro also introduces relational data semantics into existing blockchain databases and enables SQL-like queries over the data. We evaluated the performance of \acro in terms of bandwidth, latency, and computation costs. The evaluation results showed that \acro performs better when compared with the baselines.
  
  In future work, we aim to enhance the search performance by applying the indexing to the stored data and enabling the FSS keys to search over the indexed data. Moreover, we plan to explore system integration and deployment aspects in greater depth, where we will analyze the real-world implementation of \acro and its integration with current blockchain architectures.
  
\bibliographystyle{IEEEtranN}
\bibliography{IEEEabrv,References}

\end{document}